\documentclass[11pt,reqno]{amsart} 
\usepackage{amsmath, amsthm, amsfonts, amssymb} 
\usepackage[colorlinks]{hyperref} 
\usepackage[T1]{fontenc} 
\usepackage[mathscr]{eucal} 
\usepackage[polish,english]{babel}
\usepackage{graphicx}
\usepackage{float}
\newtheorem{theorem}{Theorem}[section]

\newtheorem{definition}[theorem]{Definition}

\begin{document} \title[On Thermodynamic Properties of Circumacenes Through $M(A)$]{On Thermodynamic Properties of Circumacenes Through Zagreb Indices} 
\date{} \author [S. Shafiee Alamoti and Z. Aliannejadi]{S. Shafiee Alamoti and Z. Aliannejadi} 
\begin{abstract}
We explore the fascinating world of Circumacenes, a unique family of organic molecules. Our research demonstrarates the calculation of essential topological indices, including the first and second-Zagreb indices, as well as related generalized indices, and their relationship with the structure of these molecules. We also indicate how these indices can help predict crucial thermodynamic properties, such as gap energy and electron affinity energy, using an innovative method called the Topological Indices Method $(TIM)$. 

\end{abstract} 
\maketitle \textbf{Mathematics Subject Classification (2010): 05C35; 05C07.}
\section{Introduction} 

The capacity of data restoration, data transition, and electronic component miniaturization is crucial for modern knowledge. The smaller size of electronic components not only enhances the speed of the process but also reduces energy consumption, which can be achieved through Nano-electronics. Although the electronic industry has matured technologically, industrially, and financially, the limitations of minifying electronic components on Nano-scale have led to the development of molecular electronics, a branch of Nano-technology that uses organic molecules. Aromatic hydrocarbons, such as benzene, provide suitable environments for electron transition due to p orbitals, upper and lower electron clouds, and resonance phenomenon [9,15,20,23]. Electronic circuits and logic gates can be designed by joining these hydrocarbons. One family of organic molecules that are focused on in molecular electronics is Circumacenes, which have the chemical formula of $(C _{(\frac{8i+16}{3})}H_{(\frac{2i+22}{3})})$ [13,25]. Recognizing and examining this family of Nano-structures require significant time and money. Therefore, the Topological Indices Method $(TIM)$ is a cheap and useful approach to predict electronic features [1-8,27,28]. A topological descriptor is a single number that represents a chemical structure in graph-theoretical terms via the molecular graph. If it correlates with a molecular property, it is called a topological index and is used to understand physicochemical properties of chemical compounds. Topological indices are interesting since they capture some of the properties of a molecule in a single number. Numerous topological indices have been developed and examined, with the first one introduced by Wiener. Wiener's approach involved using the sum of all shortest-path distances of a molecular graph to model the physical properties of Alkanes [17].\\

This paper focuses on connected, finite, and undirected graphs, where the degree of a vertex $u$ in a graph $A$ represented by $d_u (A)$ is the number of edges that connect it. In $1972$, Gutman and Trinajstic [18] introduced the first and second kind of Zagreb indices, which were used to examine the structure dependence of the total $\varphi$-electron energy. The first and second-Zagreb indices of a connected graph $A$ for every arbitrary adjacent vertices $u\sim v$ denoted by $M_1(A)$ and $M_2(A)$, respectively, are defined as follows:
\begin{center}
$M_1(A)=\sum_{u\sim v \in E(A)}(d_u+d_v)$
\end{center}
and
\begin{center}
$M_2(A)=\sum _{u\sim v \in E(A)}(d_ud_v)$.
\end{center}
The general sum connectivity index, $H_\alpha(A)$, is defined as [12] 
\begin{center}
$H_\alpha(A)=\sum_{u\sim v}(d_u+d_v)^\alpha$, $\hspace*{1cm}H_0(A)=|E(A)|$,
\end{center}
where $\alpha$ is an arbitrary real number. Some special cases include:\\
the first-Zagreb index, $M_1(A)=H_1(A)$ [18],\\
the sum connectivity index $SC(A)=H_{-\frac{1}{2}}(A)$ [11] and\\
the harmonic index $H(A)=2H_{-1}(A)$ [24].\\\\
The general Randi\'{c} index $R_\alpha$ of graph $A$ is a graph invariant defined as [10],
\begin{center}
$R_\alpha(A)=\sum_{u\sim v}(d_ud_v)^\alpha$, $\hspace*{1.2cm}R_0(A)=|E(A)|$,
\end{center}
where $\alpha$ is an arbitrary real number. When $\alpha=1$, then the second-Zagreb index $M_2(A)=R_1(A)$ is obtained [19] and for $\alpha=-\frac{1}{2}$ the Randi\'{c} index $R(A)=R{-\frac{1}{2}}(A)$ is obtained [22].
For $\alpha=-1$ the modified second-Zagreb index, $M_2^*(A)$, defined in [26] is obtained (see also [21]).

\begin{definition}
The sum-degree of a ring in graph $A$ is defined as the sum of total degrees of adjacent vertices in an edge of a ring, that is $\sum_{u\sim v \in R(A)} (d_u+d_v)$ and it is shown by $Sd_r(A)$.
\end{definition}
\begin{definition}
The times-degree of a ring in graph $A$ is defined as the multiplication of total degrees of adjacent vertices in an edge of a ring, that is $\sum_{u\sim v \in R(A)} (d_ud_v)$ and it is shown by $Td_r(A)$.
\end{definition}

\section{Methods and Theory}
In this section, we are going to compute the first-Zagreb index, the second-Zagreb index and some important indices for the family of Circumacenes $(C_{(\frac{8i+16}{3})}H_{(\frac{2i+22}{3})})$ molecule, (FIGURE 1).

\begin{figure}[H] 
\begin{center}
\includegraphics[scale=0.6]{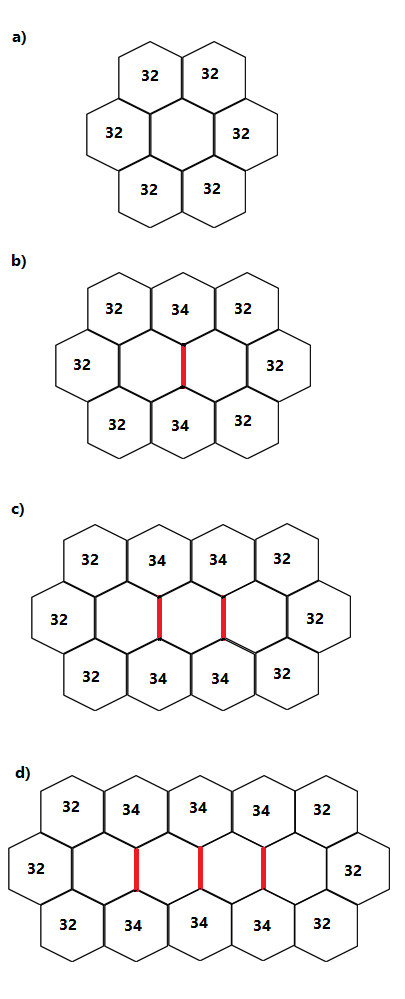}
\caption{$Linear$ $Circumacenes$ $Molecular$ $Graph$}
\label{fig:1} 
\end{center} 
\end{figure}

\begin{theorem}
If $A_i$ for $i\in \{1,2,...\}$ are simple molecular graphs of Circumacenes $(C_{(\frac{8i+16}{3})}H_{(\frac{2i+22}{3})})$, then $M_1(A_i)=62i+94$.
\end{theorem}
\begin{proof}
Consider that $A_i$ for $i\in \{1,2,...\}$ are simple molecular graphs of Circumacenes family. These graphs are formed by some outer and inner rings, in which for each graph $A_i$, $i\in \{1,2,...\}$, there are $nr_i=3i+4$ rings, in which there are $no_i=2i+4$ outer rings. We divide $E(A_i)$, $i\in \{1,2,...\}$ into three parts. The first part includes outer rings $(OR)$, the second part includes edges that do not belong to outer rings and the third part includes edges that are common among edges in outer rings. Now, we redivide the outer rings set $(OR(A_i))$, $i\in \{1,2,...\}$ into two partitions according to the sum-degree of rings $(Sd_r)$. Hence, the sets of $E(A_i)$, $i\in \{1,2,...\}$ are as follow:  
\begin{center}
$OR_1(A_i)=\{OR \in OR(A_i), i\in \{1,2,...\}| Sd_r(A_i)=32\}$,
\end{center}
\begin{center}
$OR_2(A_i)=\{OR \in OR(A_i), i\in \{1,2,...\}| Sd_r(A_i)=34\}$,
\end{center}
\begin{center}
$\acute{E}(A_i)=\{e=uv\in E(A_i)|e=uv \notin OR_1(A_i)\cup OR_2(A_i), i\in \{1,2,...\}, \sum _{u\sim v\in \acute{E}(A_i)}(d_u+d_v)=6\}$
\end{center}
and
\begin{center}
$\acute{\acute{E}}(A_i)=\{e=uv\in E(A_i)|e=uv \in \bigcap_{i=1}^n OR_1(A_i), i\in \{1,2,...\}, \sum _{u\sim v\in {\acute{\acute{E}}(A_i)}}(d_u+d_v)=6\}$
\end{center}
We start the proof by induction on the number $no_i$.\\
For $i=1$, there are $no_1=6$ outer rings that all of them belong to $OR_1(A_1)$, there are no edges in $\acute{E}(A_1)$ and there are $6$ common edges in outer rings that belong to $\acute{\acute{E}}(A_1)$ and are equal to the number of outer rings $(no_1)$. So, in graph $(A_1)$, 
\begin{center}
$M_1(A_1)=\sum_{u\sim v\in E(A_1)} (d_u+d_v)=6OR_1(A_1)+0OR_2(A_1)+0\acute{E}(A_1)-6\acute{\acute{E}}(A_1)=$
\end{center}
\begin{center}
$6\sum_{u\sim v\in OR_1(A_1)}(d_u+d_v)+0\sum_{u\sim v\in OR_2(A_1)}(d_u+d_v)+$
\end{center}
\begin{center}
$0\sum_{u\sim v \in \acute{E}(A_1)}(d_u+d_v)-6\sum_{u\sim v \in \acute{\acute{E}}(A_1)}(d_u+d_v)=$
\end{center}
\begin{center}
$6\times 32-6\times 6=156$.
\end{center}
If $i=2$, then there are $no_2=8$ outer rings in graph $A_2$. The difference in the number of outer rings in $A_1$ and $A_2$ is equal to two outer rings, so that these two additional outer rings belong to $OR_2(A_2)$. Besides, there is an edge $e=uv\in\acute{E}(A_2)$ and also $8(e=uv)\in \acute{\acute{E}}(A_2)$. Therefore, we can see easily that 
\begin{center}
$M_1(A_2)=\sum_{u\sim v\in E(A_2)} (d_u+d_v)=6OR_1(A_2)+2OR_2(A_2)+\acute{E}(A_2)-8\acute{\acute{E}}(A_2)=$
\end{center}
\begin{center}
$6\sum_{u\sim v\in OR_1(A_2)}(d_u+d_v)+2\sum_{u\sim v\in OR_2(A_2)}(d_u+d_v)+$
\end{center}
\begin{center}
$\sum_{u\sim v \in \acute{E}(A_2)}(d_u+d_v)-8\sum_{u\sim v \in \acute{\acute{E}}(A_2)}(d_u+d_v)=$
\end{center}
\begin{center}
$6\times 32+2\times 34+6-8\times 6=218$.
\end{center}
Based on FIGURE $1(c)$, if $i=3$, then there are $no_3=10$ outer rings, in which $6$ of the outer rings are in $OR_1(A_3)$ and $4$ of outer rings are in $OR_2(A_3)$. In addition, there are $2$ edges in $\acute{E}(A_3)$ and $10$ edges in $\acute{\acute{E}}(A_3)$. Thus, same as previous computations for $M_1(A_i)$, $i\in \{1,2\}$, one can get 
\begin{center}
$M_1(A_3)=\sum_{u\sim v\in E(A_3)} (d_u+d_v)=6OR_1(A_3)+4OR_2(A_3)+2\acute{E}(A_3)-10\acute{\acute{E}}(A_3)=$
\end{center}
\begin{center}
$6\sum_{u\sim v\in OR_1(A_3)}(d_u+d_v)+4\sum_{u\sim v\in OR_2(A_3)}(d_u+d_v)+$
\end{center}
\begin{center}
$2\sum_{u\sim v \in \acute{E}(A_3)}(d_u+d_v)-10\sum_{u\sim v \in \acute{\acute{E}}(A_3)}(d_u+d_v)=$
\end{center}
\begin{center}
$6\times 32+4\times 34+2\times 6-10\times 6=280$.
\end{center}
Now, let $i\geqslant 4$.\\
We assume that the claim holds for graphs $A_i$, $i\in \{4,5,...\}$. As a result we obtain:
\begin{center}
$|OR_1(A_i)|=6, |OR_2(A_i)|=2i-2, |\acute{E}(A_i)|=i-1$ and $|\acute{\acute{E}}(A_i)|=2i+4$, where $i\in \{1,2,3\}$.
\end{center}
Hence, we attain the first-Zagreb index of Circumacenes $(C_{(\frac{8i+16}{3})}H_{(\frac{2i+22}{3})})$ as follow:
\begin{center}
$M_1(A_i)=\sum_{u\sim v\in E(A_i)} (d_u+d_v)=OR_1(A_i)+OR_2(A_i)+\acute{E}(A_i)-\acute{\acute{E}}(A_i)=$
\end{center}
\begin{center}
$6\sum_{u\sim v\in OR_1(A_i)}(d_u+d_v)+(2i-2)\sum_{u\sim v\in OR_2(A_i)}(d_u+d_v)+$
\end{center}
\begin{center}
$(i-1)\sum_{u\sim v \in \acute{E}(A_i)}(d_u+d_v)-(2i+4)\sum_{u\sim v \in \acute{\acute{E}}(A_i)}(d_u+d_v)=$
\end{center}
\begin{center}
$6\times 32+(2i-2)\times 34+(i-1)\times 6-(2i+4)\times 6=62i+94=2(31i+47)$.
\end{center}
The proof is followed.
\end{proof}

\begin{theorem}
If $A_i$, $i\in \{1,2,...\}$ are simple molecular graphs of Circumacenes $(C_{(\frac{8i+16}{3})}H_{(\frac{2i+22}{3})})$, then for $i\in \{1,2,...\}$
\begin{center}
$HM(A_i)=2276i+3652$
\end{center}
\begin{center}
$SC(A_i)=H_{-\frac{1}{2}}(A_i)=-(0.066i+1.326)$
\end{center}
\begin{center}
$H(A_i)=2H_{-1}(A_i)=-2(0.108i+0.702)$.
\end{center}
\end{theorem}
\begin{proof}
As we saw in Theorem 2.1, for $i\in \{1,2,...\}$ 
\begin{center}
$M_1(A_i)=6OR_1(A_i)+(2i-2)OR_2(A_i)+(i-1)\acute{E}(A_i)-(2i+4)\acute{\acute{E}}(A_i)$.
\end{center}
To calculate $HM(A_i)$, $SC(A_i)$ and $H(A_i)$, $i\in \{1,2,...\}$, it is enough to replace $(d_u+d_v)^2$, $(d_u+d_v)^{-\frac{1}{2}}$ and $(d_u+d_v)^{-1}$ instead of $(d_u+d_v)$. Clearly, it is proved for $i\in \{1,2,...\}$. Therefore, 
\begin{center}
$HM(A_i)=2276i+3652$
\end{center}
\begin{center}
$SC(A_i)=H_{-\frac{1}{2}}(A_i)=-(0.066i+1.326)$
\end{center}
\begin{center}
$H(A_i)=2H_{-1}(A_i)=-2(0.108i+0.702)$.
\end{center}
\end{proof}
\begin{theorem}
If $A_i$, $i\in \{1,2,...\}$ are simple molecular graphs of Circumacenes $(C_{(\frac{8i+16}{3})}H_{(\frac{2i+22}{3})})$, then $M_2(A_i)=3(29i+66)$.
\end{theorem}
\begin{proof}
By inspiration of the proof of Theorem 2.1, we divide edge sets $E(A_i)$ and outer ring sets $OR(A_i)$, $i\in \{1,2,...\}$ into four partitions according to times-degree of rings $(Td_r)$. Thus, the sets of $E(A_i)$ and $OR(A_i)$, $i\in \{1,2,...\}$ are as follow:
\begin{center}
$OR_1^*(A_i)=\{OR\in R(A_i), i\in \{1,2,...\}| Tdr(A_i)=34\}$,
\end{center}
\begin{center}
$OR_2^*(A_i)=\{OR\in R(A_i), i\in \{1,2,...\}| Tdr(A_i)=48\}$,
\end{center}
\begin{center}
$\acute{E}^*(A_i)=\{e=uv\in E(A_i)|e=uv \notin OR_1^*(A_i)\cup OR_2^*(A_i), i\in \{1,2,...\}, \sum _{u\sim v\in \acute{E}^*(A_i)}(d_ud_v)=9\}$
\end{center}
and
\begin{center}
$\acute{\acute{E}}^*(A_i)=\{e=uv\in E(A_i)|e=uv \in \bigcap_{i=1}^n OR_1^*(A_i), i\in \{1,2,...\}, \sum _{u\sim v\in {\acute{\acute{E}}^*(A_i)}}(d_ud_v)=9\}$.
\end{center}
In a completely similar way to the proof of Theorem 2.1, there are 
\begin{center}
$|OR_1^*(A_i)|=6, |OR_2^*(A_i)|=2i-2, |\acute{E}^*(A_i)|=i-1$ and $|\acute{\acute{E}}^*(A_i)|=2i+4$, where $i\in \{1,2,3\}$.
\end{center}
So,
\begin{center}
$M_2(A_i)=\sum_{u\sim v\in E(A_i)}(d_ud_v)=$
\end{center}
\begin{center}
$OR_1^*(A_i)+OR_2^*(A_i)+\acute{E}^*(A_i)-\acute{\acute{E}}^*(A_i)=$
\end{center}
\begin{center}
$6\sum_{u\sim v\in OR_1^*(A_i)}(d_ud_v)+(2i-2)\sum_{u\sim v\in OR_2^*(A_i)}(d_ud_v)+$
\end{center}
\begin{center}
$(i-1)\sum_{u\sim v \in \acute{E}^*(A_i)}(d_ud_v)-(2i+4)\sum_{u\sim v \in \acute{\acute{E}}^*(A_i)}(d_ud_v)=$
\end{center}
\begin{center}
$6\times 43+(2i-2)\times 48+(i-1)\times 9-(2i+4)\times 9=87i+198=3(29i+66)$.
\end{center}
The proof is completed.
\end{proof}
\begin{theorem}
If $A_i$, $i\in \{1,2,...\}$ are simple molecular graphs of Circumacenes $(C_{(\frac{8i+16}{3})}H_{(\frac{2i+22}{3})})$, then for $i\in \{1,2,...\}$
\begin{center}
$R(A_i)=R_{-\frac{1}{2}}(A_i)=-(0.045i+0.897)$,
\end{center}
\begin{center}
$M_2^*(A_i)=R_{-1}(A_i)=-(0.071i+0.437)$.
\end{center}
\end{theorem}
\begin{proof}
According to the proof of Theorem 2.3, for $i\in\{1,2,...\}$, $M_2(A_i)=$
\begin{center}
$6\sum_{u\sim v\in OR_1^*(A_i)}(d_ud_v)+(2i-2)\sum_{u\sim v\in OR_2^*(A_i)}(d_ud_v)+$
\end{center}
\begin{center}
$(i-1)\sum_{u\sim v \in \acute{E}^*(A_i)}(d_ud_v)-(2i+4)\sum_{u\sim v \in \acute{\acute{E}}^*(A_i)}(d_ud_v)$.
\end{center}
To calculate $R(A_i)$ and $M_2(A_i)$, $i\in\{1,2,...\}$, we replace $(d_ud_v)^{-\frac{1}{2}}$ and $(d_ud_v)^{-1}$ instead of $(d_ud_v)$. Obviously, it is investigated for $i\in\{1,2,...\}$. Therefore, 
\begin{center}
$R(A_i)=R_{-\frac{1}{2}}(A_i)=-(0.045i+0.897)$,
\end{center}
\begin{center}
$M_2^*(A_i)=R_{-1}(A_i)=-(0.071i+0.437)$.
\end{center}
\end{proof}
 
 \section{Results}
 
Equation $M_1(A_i)=62i+94$, $i\in \{1,2,...\}$ was used to calculate the first-Zagreb index for the five members of the Circumacenes, which are then displayed in TABLE 1.\\

\begin{table}[H]
\centering
\caption{$M_1$ index for the first five members of Circumacenes family}
\begin{tabular}{|c|c|c|c|}
\hline
Chemical Formula & IUPAC Name & Rings Number $(n_i)$ & first-Zagreb index $(M_1)$\\
\hline
$C_{24}H_{12}$ & Coronene & $7$ & $528$\\
\hline
$C_{32}H_{14}$ & Ovalene & $10$ & $714$\\
\hline
$C_{40}H_{16}$ & Circumanthracene & $13$ & $900$\\
\hline
$C_{48}H_{18}$ & Circumtetracene & $16$ & $1086$\\
\hline
$C_{56}H_{20}$ & Circumpentacene & $19$ & $1272$\\
\hline
\end{tabular}
\end{table}
We conducted calculations on Electro-optical characteristics like gap energy $(E_g)$ and electron affinity energy $(E_{ea})$ in the Circumacenes family using empirical data to examine the thermodynamic properties. These results were then compared against established sources, as detailed in TABLE 2 [14,16]. 
\begin{table}[H]
\centering
\caption{Gap and Electron affinity energy of the first five members of Circumacenes family}
\begin{tabular}{|c|c|c|}
\hline
Chemical Formula & $E_g$ & $E_{ea}$\\
\hline
$C_{24}H_{12}$ & $3.24$ & $0.96$\\
\hline
$C_{32}H_{14}$ & $2.69$ & $1.55$\\
\hline
$C_{40}H_{16}$ & $2.03$ & $1.91$\\
\hline
$C_{48}H_{18}$ & $1.54$ & $2.17$\\
\hline
$C_{56}H_{20}$ & $1.18$ & $2.37$\\
\hline
\end{tabular}
\end{table}
FIGURES 1, 2 illustrate the values of $E_g$ and $E_{ea}$ for the Circumacenes family about the first-Zagreb index, reflecting the findings of this study. FIGURES 1 and 2 demonstrate the precise calculation of $E_g$ and $E_{ea}$ using the $M_1(A)$ index variations within the Circumacenes family.

\begin{figure}[H]
\centering
\includegraphics[scale=0.5]{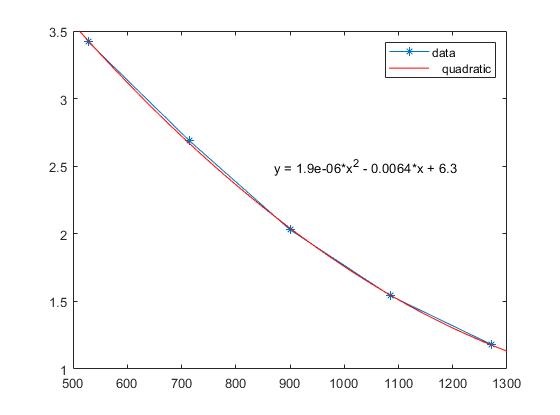}
\caption{Gap energy changes in Circumacenes family according to $M_1$ index}
\label{fig:1}       
\end{figure}
\begin{figure}[H]
\centering
\includegraphics[scale=0.5]{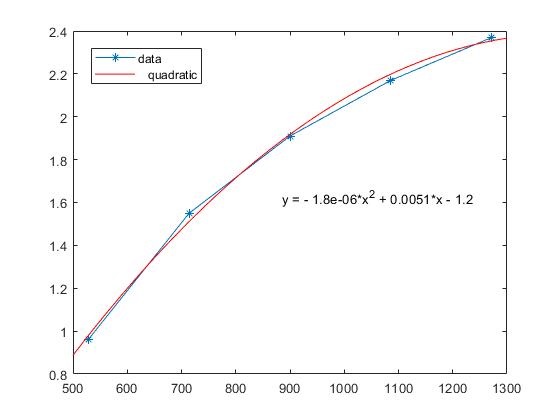}
\caption{Electron affinity energy changes in Circumacenes family according to $M_1$ index}
\label{fig:1}       
\end{figure}

These forecasts can be made using a pair of equations $(3-1)$ and $(3-2)$: 
\begin{center}
$E_{ea}=-1.8e - 06(M_1)^2 +0.0051 M_1 - 1.2$ $\hspace*{2cm}(3-1)$,       
\end{center}
\begin{center}
$E_g=1.9e - 06(M_1)^2 - 0.0064 M_1 +6.3$ $\hspace*{2.5cm}(3-2)$. 
\end{center}

\section{Conclusion}
In TABLE 3, we present the values of $E_G$, $E_{ea}$, for the Circumacenes family using the $TIM$ method (equations $(3-1)$ and $(3-2)$). This table highlights the reliability of the $TIM$ method by comparing the results with the reference values from TABLE 2. 

\begin{table}[H]
\centering
\caption{Calculation gap and electron affinity energy of Circumacenes family through $TIM$ method}
\begin{tabular}{|c|c|c|}
\hline
Chemical Formula & $E_g$ & $E_{ea}$\\
\hline
$C_{24}H_{12}$ & $3.4504896$ & $0.9909888$\\
\hline
$C_{32}H_{14}$ & $2.6990124$ & $1.5237672$\\
\hline
$C_{40}H_{16}$ & $2.079$ & $1.932$\\
\hline
$C_{48}H_{18}$ & $1.5904524$ & $2.2156872$\\
\hline
$C_{56}H_{20}$ & $1.2333696$ & $2.3477478$\\
\hline
\end{tabular}
\end{table}
 
The $TIM$ method, along with the first-Zagreb index, allows for accurate prediction of various physical and chemical properties within the Circumacenes family, saving time and resources compared to traditional theoretical and experimental approaches that often yield only approximate results. Moving forward, several heavier members of the Circumacenes family will undergo analysis using the $TIM$ method. Equations $(3-1)$ and $(3-2)$ have been employed to forecast the energies of $E_g$ and $E_{ea}$ with the outcomes detailed in TABLE 4. 
\begin{table}[H]
\centering
\caption{Prediction gap and electron affinity energy of Circumacenes family through $TIM$ method}
\begin{tabular}{|c|c|c|c|}
\hline
Chemical Formula & $M_1$ index & $E_{ea}$ & $E_g$\\
\hline
$C_{24}H_{12}$ & 1458 & $2.4094248$ & $1.0077516$\\
\hline
$C_{32}H_{14}$ & 1644 & $2.3194752$ & $0.9135984$\\
\hline
$C_{40}H_{16}$ & 1830 & $2.10498$ & $0.95091$\\
\hline
$C_{48}H_{18}$ & 2016 & $1.7659392$ & $1.1196864$\\
\hline
$C_{56}H_{20}$ & 2202 & $1.3023528$ & $1.4199279$\\
\hline
\end{tabular}
\end{table}

Somayeh Shafiee Alamoti \\ Department of Mathematics, Islamic Azad University, South Tehran Branch, \\ Tehran, Iran. \\ e-mail: shafiee.s88@gmail.com \\ Phone: +989126994492
 \bigskip \noindent 
 
Zohreh Aliannejadi \\ Department of Mathematics, Islamic Azad University, South Tehran Branch, \\ Tehran, Iran. \\ e-mail: z\_alian@azad.ac.ir \\ phone: +989122797710
 \bigskip \noindent 

 \end{document}